\documentclass[english]{article}
\usepackage[T1]{fontenc}
\usepackage[latin9]{inputenc}
\usepackage{amsmath,amsthm} 
\usepackage{graphicx}
\usepackage{amssymb}
\usepackage{enumitem}
\usepackage{natbib} 
\setcitestyle{round}
\usepackage{bussproofs}
\usepackage[letterpaper]{geometry} 
\usepackage{appendix}

\usepackage{thmtools,thm-restate}

\usepackage{siunitx}

\usepackage{tikz}
\usepackage{pgf} 
\usetikzlibrary{patterns,automata,arrows,shapes,snakes,topaths,trees,backgrounds,positioning,through,calc}
\usepackage{setspace}
\usepackage{multicol}
\usepackage{graphicx}

\theoremstyle{definition}
\newtheorem{theorem}{Theorem}

\newtheorem{definition}[theorem]{Definition}
\newtheorem{lemma}[theorem]{Lemma}

\makeatletter  

\usepackage{stmaryrd}  
\usepackage{comment}
\usepackage{hyperref}
\hypersetup{colorlinks=true, citecolor=blue, linkcolor=blue}  

\makeatother

\usepackage{babel} 
\onehalfspace
 
\begin{document} 

\title{An impossibility theorem concerning \\ positive involvement in voting}

\author{Wesley H. Holliday \\ University of California, Berkeley}

\date{{\normalsize Published in \textit{Economics Letters}, Vol.~236, 2024, 111589.}}

\maketitle

\begin{abstract} In social choice theory with ordinal preferences, a voting method satisfies the axiom of \textit{positive involvement} if adding to a preference profile a voter who ranks an alternative uniquely first cannot cause that alternative to go from winning to losing. In this note, we prove a new impossibility theorem concerning this axiom: there is no ordinal voting method satisfying positive involvement that also satisfies the Condorcet winner and loser criteria, resolvability, and a common invariance property for Condorcet methods, namely that the choice of winners depends only on the ordering of majority margins by size.\end{abstract}

\section{Introduction}\label{Intro}

A basic assumption of democratic voting is that increased voter support for a candidate should not harm that candidate's chances of winning. In social choice theory with ordinal preferences, one formalization of this idea is given by the axiom of \textit{positive involvement} (\citealt{Saari1995}): adding to a preference profile a voter who ranks an alternative uniquely first cannot cause that alternative to go from winning to losing.\footnote{Of course, such a voter expresses more than just support for their favorite alternative if they also rank other alternatives. A weaker formalization is therefore \textit{bullet-vote positive involvement}: adding a voter who ranks an alternative uniquely first and all other alternatives in a tie below their favorite cannot cause their favorite to go from winning to losing.} Perhaps surprisingly, many well-known voting methods---especially Condorcet-consistent voting methods---violate this axiom (see \citealt{Perez1995}). In this note, we prove a new impossibility theorem in this vein: there is no ordinal voting method satisfying positive involvement that also satisfies the Condorcet winner and loser criteria (\citealt{Condorcet1785}), resolvability  (\citealt{Tideman1986}), and a common invariance property for Condorcet methods, namely that the choice of winners depends only on the ordering of majority margins by size. Methods satisfying this \textit{ordinal margin invariance} include Minimax (\citealt{Simpson1969}, \citealt{Kramer1977}), Ranked~Pairs (\citealt{Tideman1987}), Beat Path (\citealt{Schulze2011}), Weighted Covering (\citealt{Dutta1999}, \citealt{Fernandez2018}), Split Cycle (\citealt{HP2023}), and Stable Voting (\citealt{HP2023b}). Thus, our theorem helps explain why none of these voting methods satisfies all the stated axioms, as well as motivating the search for new methods satisfying all the axioms except for ordinal margin invariance.

\section{Preliminaries}

Fix infinite sets $\mathcal{X}$ and $\mathcal{V}$ of \textit{alternatives} and \textit{voters}, respectively. A \textit{profile} $\mathbf{P}$ is a function  from some nonempty finite $V(\mathbf{P})\subseteq\mathcal{V}$, called the set of \textit{voters in $\mathbf{P}$}, to the set of strict weak orders on some nonempty finite $X(\mathbf{P})\subseteq\mathcal{X}$, called the set of \textit{alternatives in $\mathbf{P}$}.\footnote{The proof of Theorem \ref{MainThm} easily adapts to a setting where the set of alternatives cannot vary between profiles, assuming at least four alternatives (for more, add indefensible ones below the four in the proof). But we need the set of voters to be variable.} For $i\in V(\mathbf{P})$ and $x,y\in X(\mathbf{P})$,  take $(x,y)\in\mathbf{P}(i)$ to mean that  $i$ strictly prefers  $x$ to  $y$.  $\mathbf{P}$ is \textit{linear} if for each $i\in V(\mathbf{P})$,  $\mathbf{P}(i)$ is a linear order.\footnote{Theorem \ref{MainThm} holds even if we restrict the domain of voting methods to linear profiles, but we allow non-linear profiles in our setup for the sake of the distinction between the two parts of Lemma \ref{RefineLem}.} Given $x,y\in X(\mathbf{P})$, the \textit{margin of $x$ over $y$} is defined by
\[\text{Margin}_\mathbf{P}(x,y) = \text{Support}_\mathbf{P}(x,y) -  \text{Support}_\mathbf{P}(y,x) \mbox{, where } \text{Support}_\mathbf{P}(a,b) =|\{i\in V(\mathbf{P})\mid (a,b)\in \mathbf{P}(i)\}| .\]
A \textit{voting method} is a function $F$ assigning to each profile $\mathbf{P}$ a nonempty $F(\mathbf{P})\subseteq X(\mathbf{P})$.\footnote{Thus, we build the axiom of \textit{universal domain} into the definition of our voting methods.} If $|F(\mathbf{P})|>1$, we assume some further (possibly random) tiebreaking process ultimately narrows $F(\mathbf{P})$ down to one alternative. A voting method $G$ \textit{refines} $F$ (for some class of profiles) if for any profile $\mathbf{P}$ (in the class), $G(\mathbf{P})\subseteq F(\mathbf{P})$. 

A voting method $F$ satisfies  \textit{positive involvement} (see~\citealt{Saari1995}, \citealt{Perez2001}, \citealt{HP2021b}) if for any profiles $\mathbf{P},\mathbf{P}'$, if $x\in F(\mathbf{P})$ and $\mathbf{P}'$ is obtained from $\mathbf{P}$ by adding one voter who ranks $x$ uniquely in first place, then $x\in F(\mathbf{P}')$.\footnote{The requirement that the new voter ranks $x$ \textit{uniquely} first is crucial; a stronger version of the axiom that applies whenever the new voter does not strictly prefer any $y$ to $x$ is inconsistent with the Condorcet winner criterion (\citealt{Perez2001}, \citealt{Duddy2014}).} Thus, you cannot cause $x$ to lose by ranking $x$ first. According to P\'{e}rez \citeyearpar[p.~605]{Perez2001}, this ``may be seen as the minimum to require concerning the coherence in the winning set when new voters~are~added.''

A voting method $F$ satisfies the \textit{Condorcet winner criterion} (\citealt{Condorcet1785}) if for every profile $\mathbf{P}$ with a Condorcet winner, $F(\mathbf{P})$ contains only the Condorcet winner. Recall that $x$ is a \textit{Condorcet winner} (resp.~\textit{weak Condorcet winner}) if $x$ beats (resp.~does not lose to) any other alternative head-to-head, i.e., for every $y\in X(\mathbf{P})\setminus\{x\}$,  $\text{Margin}_\mathbf{P}(x,y)>0$ (resp.~$\text{Margin}_\mathbf{P}(x,y)\geq 0$).  $F$ satisfies the \textit{weak Condorcet winner criterion} if for every profile $\mathbf{P}$ with a weak Condorcet winner, $F(\mathbf{P})$ contains only weak Condorcet winners. 

A key tool for reasoning about voting methods satisfying both positive involvement and the Condorcet winner criterion is given by the following voting method, appearing implicitly in Lemma 3 of \citealt{Perez1995}  and studied explicitly in \citealt{Kasper2019}.

\begin{definition}\label{DefenseDef} Given a profile $\mathbf{P}$, define the \textit{defensible set of $\mathbf{P}$} as
\[D(\mathbf{P})=\{x\in X(\mathbf{P})\mid \mbox{for all } y\in X(\mathbf{P})\mbox{, there exists }z\in X(\mathbf{P}): \text{Margin}_\mathbf{P}(z,y)\geq \text{Margin}_\mathbf{P}(y,x)\}.\]
\end{definition}
\noindent As a voting method, $D$ is a special case, based on majority margins, of Heitzig's \citeyearpar{Heitzig2002} family of methods that are \textit{weakly immune to binary arguments}.\footnote{In particular, $D$ is the coarsest voting method satisfying all of Heitzig's axioms $\mathrm{wIm}_{P_\alpha}$ for $\alpha\in (0,1]$.} The idea can also be found in the argumentation theory of Dung \citeyearpar{Dung1995}.\footnote{In Dung's \citeyearpar[Definition~6.1]{Dung1995} terms, the alternatives in $D(\mathbf{P})$ are \textit{acceptable} with respect to $X(\mathbf{P})$ given each attack relation $attacks_n$, for each positive integer $n$, defined by $attacks_n(x,y)$ if $\text{Margin}_\mathbf{P}(x,y)\geq n$.} Our choice of the term  `defensible' is based on the following intuition. Although $x$ may lose head-to-head to $y$, possibly prompting the supporters of $y$ over $x$ to call for the overthrow of $x$ in favor of $y$, we can defend the choice of $x$ on the following grounds, when applicable: there is another $z$ that beats $y$  head-to-head by a margin at least as large as that by which $y$ beat $x$; thus, the reason invoked to overthrow $x$ in favor of $y$ would immediately provide at least as strong a reason to overthrow $y$ in favor of $z$, undercutting the original idea to overthrow $x$ in favor of $y$. 

Part \ref{RefineLem1} of the following is equivalent to Lemma 3 of \citealt{Perez1995} (cf.~Lemma 1 of \citealt{Perez2001}),\footnote{For linear profiles, where $\text{Support}_\mathbf{P}(x,y) + \text{Support}_\mathbf{P}(y,x)=|V(\mathbf{P})|$, one can show that $\text{Margin}_\mathbf{P}(z,y)\geq \text{Margin}_\mathbf{P}(y,x)$ is equivalent to $\text{Support}(x,y)\geq \text{Support}(y,z)$, which is the relevant condition in Lemma 3 of \citealt{Perez1995}.} which is in turn an adaptation of a similar result of \citealt{Moulin1988}, but we include a proof to keep this note self-contained.

\begin{lemma}\label{RefineLem} $\,$
\begin{enumerate}
\item\label{RefineLem1} Any voting method satisfying positive involvement and the Condorcet winner criterion refines the defensible set for linear profiles.
\item\label{RefineLem2} Any voting method satisfying positive involvement and the weak Condorcet winner criterion refines the defensible set for all profiles.
\end{enumerate}
\end{lemma}

\begin{proof} For part \ref{RefineLem1}, let $F$ be such a method, $\mathbf{P}$ a linear profile, and $x\in F(\mathbf{P})$. Toward a contradiction, suppose $x\not\in D(\mathbf{P})$. Hence there is a $y\in X(\mathbf{P})$ with $\text{Margin}_\mathbf{P}(y,x)>0$ and for every $z\in X(\mathbf{P})$, $\text{Margin}_\mathbf{P}(z,y)< \text{Margin}_\mathbf{P}(y,x)$. Let $k= max\{\text{Margin}_\mathbf{P}(z,y) \mid z\in X(\mathbf{P})\}$, so $k<\text{Margin}_\mathbf{P}(y,x)$. Since $\mathbf{P}$ is linear, all margins have the same parity, so $k+1<\text{Margin}_\mathbf{P}(y,x)$. Let $\mathbf{P}'$ be obtained from $\mathbf{P}$ by adding $k+1$ voters who rank $x$ uniquely first and $y$ uniquely second, followed by any linear order of the remaining alternatives. Then by positive involvement, $x\in F(\mathbf{P}')$. But in $\mathbf{P}'$, $y$ is the Condorcet winner, so by the Condorcet winner criterion,  $F(\mathbf{P}')=\{y\}$, contradicting $x\in F(\mathbf{P}')$. Thus, $x\in D(\mathbf{P})$.

For part \ref{RefineLem2} and an arbitrary  $\mathbf{P}$, the argument is similar only the margins may have different parities, so we cannot infer from $k< \text{Margin}_\mathbf{P}(y,x)$ that  $k+1<\text{Margin}_\mathbf{P}(y,x)$. In this case, we add only $k$ voters who rank $x$ uniquely first and $y$ uniquely second, followed by any linear order of the remaining alternatives. Then in $\mathbf{P}'$, $y$ is a \textit{weak} Condorcet winner, and $\text{Margin}_{\mathbf{P}'}(y,x)>0$. It follows by the weak Condorcet winner criterion that $x\not\in F(\mathbf{P}')$, contradicting positive involvement.\end{proof}

Given the appeal of positive involvement and the Condorcet winner criterion, it is of significant interest to explore refinements of the defensible set (cf.~\citealt[\S~4]{Kasper2019}).  The Minimax and Split Cycle methods satisfy these axioms (\citealt{Perez2001}, \citealt{HP2023}), so they refine the defensible set.\footnote{This is also clear directly. Minimax selects the alternatives whose worst head-to-head loss is smallest among any alternative's worst head-to-head loss. Thus, if $x$ is a Minimax winner that loses head-to-head to $y$, there must be an alternative $z$ to which $y$ loses by at least as large a margin. As noted by Young \citeyearpar{Young1977}, the Minimax winner is the ``\textit{most stable against overthrow}'' (p.~350). If $x$ is a Split Cycle winner that loses head-to-head to $y$, then there is a sequence of alternatives $y_1,\dots,y_n$ with $y=y_1$ and $x=y_n$ such that each  loses to the next by a margin at least as large as that by which $x$ lost to $y$; hence the reason for overthrowing $x$ in favor of $y$ would similarly justify a sequence of revolutions \textit{leading right back to $x$}.} However, refinements of the defensible set may violate positive involvement. Examples include some refinements of Split Cycle such as Beat Path, Ranked~Pairs, and Stable Voting~(see \citealt{HP2023}).

\section{Impossibility}\label{Imposs}

Let us now try adding further axioms to our wish list. A \textit{Condorcet loser} in a profile $\mathbf{P}$ is an $x\in X(\mathbf{P})$ that loses head-to-head to every other alternative, i.e., for every $y\in X(\mathbf{P})\setminus \{x\}$,  $\text{Margin}_\mathbf{P}(y,x)>0$.  $F$ satisfies the \textit{Condorcet loser criterion} if for every profile $\mathbf{P}$, $F(\mathbf{P})$ does not contain a Condorcet loser. 

The \textit{resolvability} axiom can be formulated in two ways, either of which works below. The first formulation (from \citealt{Tideman1986}), \textit{single-voter resolvability}, states that for any profile $\mathbf{P}$, if $F(\mathbf{P})$ contains multiple alternatives, then there is a $\mathbf{P}'$ obtained from $\mathbf{P}$ by adding only one new voter such that $F(\mathbf{P}')$ contains only one alternative; thus, every tie can be broken by adding just one voter. The second formulation (see \citealt[\S~4.2.1]{Schulze2011}),  \textit{asymptotic resolvability}, states that for any positive integer $m$, in the limit as the number of voters goes to infinity, the proportion of linear profiles $\mathbf{P}$ for $m$ alternatives with $|F(\mathbf{P})|>1$  goes to zero; thus, ties become vanishingly rare.

The final axiom is \textit{ordinal margin invariance}. Informally, $F$ satisfies this axiom if its selection of alternatives  depends only on the ordering of majority margins by size, not on the absolute margins or other features of the profile. Formally, given a profile $\mathbf{P}$, define $\mathbb{M}(\mathbf{P})=(M,\succ)$, the \textit{ordinal margin graph of $\mathbf{P}$}, where $M$ is a directed graph whose set of vertices is $X(\mathbf{P})$ with an edge from $x$ to $y$ when $\text{Margin}_\mathbf{P}(x,y)>0$, and $\succ$ is a strict weak order of the edges of $M$ such that $(a,b)\succ (c,d)$ if $\text{Margin}_\mathbf{P}(a,b)>\text{Margin}_\mathbf{P}(c,d)$. Then $F$ satisfies ordinal margin invariance if for any $\mathbf{P}$,  $\mathbf{P}'$, if $\mathbb{M}(\mathbf{P})=\mathbb{M}(\mathbf{P}')$, then $F(\mathbf{P})=F(\mathbf{P}')$. As a corollary of Debord's Theorem (\citealt{Debord1987}, cf.~\citealt[Theorem 4.1]{Fischer2016}), any pair $\mathbb{M}=(M,\succ)$ of an asymmetric directed graph and a strict weak order of its edges is the ordinal margin graph of some profile. Hence if $F$ satisfies ordinal margin invariance, we may also regard $F$ as a function that takes in $\mathbb{M}=(M,\succ)$ and returns a nonempty subset $F(\mathbb{M})$ of its vertices.

While the normative appeal of positive involvement and the Condorcet criteria, as well as the practical relevance of resolvability, is clear, it is less obvious whether ordinal margin invariance should be a desideratum.  One point in favor of methods satisfying ordinal margin invariance is that a small amount of noise in the collection of voter preferences is unlikely to change the ordinal margin graph and therefore unlikely to change the selection of winners, rendering such rules quite robust to noise (cf.~\citealt{Procaccia2006}).

For the following, a \textit{linearly edge-ordered tournament} is a pair $(M,\succ)$ where $M$ is a tournament, i.e., an asymmetric directed graph in which any two distinct vertices are related by an edge in some direction, and $\succ$ is a linear order of the tournament's edges.

\begin{lemma}\label{ResolveLem} If $F$ satisfies ordinal margin invariance and single-voter (resp.~asymptotic) resolvability, then $F$ selects a unique winner in any profile whose ordinal margin graph is a linearly edged-ordered tournament.
\end{lemma}
 
\begin{proof} For single-voter resolvability, suppose $F$ selects multiple winners in a profile $\mathbf{P}$ whose ordinal margin graph is a linearly edge-ordered tournament. Consider $3\mathbf{P}$, the profile obtained from $\mathbf{P}$ by replacing each voter with three copies of that voter. Since $3\mathbf{P}$ has the same ordinal margin graph as $\mathbf{P}$, $F$ selects the same  winners in $3\mathbf{P}$ as in $\mathbf{P}$. But by adding a single voter to $3\mathbf{P}$, it is impossible to obtain a $\mathbf{P}'$ whose ordinal margin graph differs from that of $3\mathbf{P}$, so it is impossible to obtain a profile with a unique winner. Hence $F$ does not satisfy single-voter resolvability.

For asymptotic resolvability, Harrison-Trainor \citeyearpar[Theorem~11.2]{HT2020} shows that for every positive integer $m$ and linearly edge-ordered tournament $T$ for $m$ alternatives, in the limit as the number of voters goes to infinity, the proportion of linear profiles for $m$ alternatives whose ordinal margin graph is $T$ is nonzero. Thus, if $F$ picks multiple winners in such a tournament, $F$ does not satisfy asymptotic resolvability.\end{proof}

Thanks to Lemma \ref{ResolveLem}, either version of resolvability works below, so we simply speak of `resolvability'. The defensible set does not satisfy either version, since it returns multiple winners for some linearly edge-ordered tournaments. Table \ref{TournTable} shows the extent of its irresoluteness for linearly edge-ordered tournaments for four alternatives in comparison to several other voting methods.\footnote{The Smith set (\citealt{Smith1973}) is the smallest nonempty set of alternatives such that every alternative in the set beats ever alternative outside the set head-to-head. The uncovered set (\citealt{Fishburn1977}, \citealt{Miller1980}) is the set of all alternatives $x$ for which there is no $y$ such that $y$ beats $x$ and beats every $z$ that $x$ beats; this definition is equivalent to others (see \citealt{Duggan2013}) provided there are no zero margins between distinct alternatives, which is the case for Table \ref{TournTable}. In this case, the Copeland winners (\citealt{Copeland1951}) are the alternatives that beat the most other alternatives. A notebook with code to generate Table \ref{TournTable} and Figures \ref{UpperLeftToMiddleProfs}--\ref{LowerRightToMiddleProfs} below is available at \href{https://github.com/wesholliday/pos-inv}{github.com/wesholliday/pos-inv}.}

\begin{table}[h]
\begin{center}
\begin{tabular}{lccc}
 & \# with multiple winners & avg.~size of set & max.~size of set\\
Smith set & 960 & 2.375\phantom{111} & 4  \\
Uncovered set & 960 & 2\phantom{.000000} & 3  \\
Copeland  & 960 & 1.625\phantom{111} & 3 \\ 
Defensible set & 598 & 1.359375 & 3  \\
Defensible set $\cap$ Smith set & 583 & 1.34375\phantom{1} & 3 \\
Split Cycle & 104 & $1.0541\overline{6}\phantom{1}$ & 2  \\
Minimax & \phantom{11}0& 1\phantom{.000000} & 1 
\end{tabular}

\end{center}
\caption{Among the 1,920  linearly edge-ordered tournaments for four alternatives up to isomorphism, the number with multiple winners for a given method and the average (resp.~maximum) size of the set of winners.}\label{TournTable}
\end{table}

If we drop any of the axioms introduced so far besides the Condorcet winner criterion and ordinal margin invariance, then  there is a method satisfying the remaining axioms. This is shown by  Beat Path/Ranked Pairs/Stable Voting,  Minimax,\footnote{It is plausible that with more combinatorial work, our proof strategy for Theorem \ref{MainThm} could yield a characterization of Minimax for four alternatives using the axioms other than the Condorcet loser criterion~(cf.~\citealt{HP2023c}).} and Split Cycle in Table \ref{IndTable}, which also includes the Borda method (\citealt{Borda1781}) and Black's method\footnote{Black's method selects the Condorcet winner only, if one exists, and otherwise selects all Borda winners.} (\citealt{Black1958}) for comparison. However, no method satisfies all of the axioms.\footnote{In Theorem \ref{MainThm}, positive involvement can be replaced by \textit{negative involvement} (see \citealt{Perez2001}) by the proof of Proposition~3.19 in \citealt{Ding2023}, since ordinal margin invariance implies the neutral reversal axiom in that proposition.}

\begin{table}
\begin{center}
\begin{tabular}{lcccccc}
 					& Beat Path & Black's & Borda & Minimax  & Split Cycle \\
					& Ranked Pairs \\
					& Stable Voting \\
Positive involvement 		& $-$		&$-$ & $\checkmark$ & $\checkmark$ & $\checkmark$    \\
Condorcet winner 		& $\checkmark$  & $\checkmark$ & $-$ &  $\checkmark$ &   $\checkmark$  \\
Condorcet loser 		& $\checkmark$  & $\checkmark$ & $\checkmark$ & $-$ &   $\checkmark$  \\
Resolvability 			& $\checkmark$  &$\checkmark$ & $\checkmark$ &  $\checkmark$ &  $-$  \\
Ordinal margin invariance 	& $\checkmark$  & $-$ & $-$ &  $\checkmark$ & $\checkmark$
\end{tabular} 
\end{center}
\caption{Satisfaction ($\checkmark$) or violation ($-$) of the axioms by selected voting methods.}\label{IndTable}
\end{table}

\begin{theorem}\label{MainThm} There is no voting method satisfying positive involvement, the Condorcet winner criterion, the Condorcet loser criterion, resolvability, and ordinal margin invariance.
\end{theorem}

\begin{proof} Assume there is such an $F$.  To derive a contradiction, we use the ordinal margin graphs in Figure~\ref{RelMarginsFig}. The numbers  indicate the ordering $\succ$ from the smallest margin ($1$) to the largest  ($6$). The defensible set for each graph is shaded in gray. Since $F$ satisfies ordinal margin invariance and resolvability, $F$ returns a singleton for each graph by Lemma~\ref{ResolveLem}. Given $\mathbb{M}$ and $\mathbb{M}'$ and an alternative $x$, we write $\mathbb{M}\Rightarrow_x\mathbb{M}'$ in Figure~\ref{RelMarginsFig} if there are profiles $\mathbf{P},\mathbf{P}'$ such that $\mathbb{M}$ is the ordinal margin graph of $\mathbf{P}$, $\mathbb{M}'$ is the ordinal margin graph of $\mathbf{P}'$, and $\mathbf{P}'$ is obtained from $\mathbf{P}$ by adding voters all of whom rank $x$ uniquely first. The construction of such profiles is an integer linear programming problem, whose solution yields the profiles in Figures \ref{UpperLeftToMiddleProfs}--\ref{LowerRightToMiddleProfs}.\footnote{Although we minimize the number of voters in these profiles subject to the relevant constraints, this does not answer an interesting question: what is the minimal number of voters needed for the impossibility theorem itself (cf.~\citealt{Brandt2017b})?}

The defensible set for $\mathbb{M}_1$ is $\{a,d\}$. Then since $F$ satisfies positive involvement and the Condorcet winner criterion, $F(\mathbb{M}_1)\subseteq\{a,d\}$ by Lemma \ref{RefineLem}, so $F(\mathbb{M}_1)=\{a\}$ or $F(\mathbb{M}_1)=\{d\}$ by resolvability.

Suppose $F(\mathbb{M}_1)=\{d\}$. Then by positive involvement and resolvability, $F(\mathbb{M}_2)=\{d\}$. On the other hand,  $F(\mathbb{M}_3)\subseteq \{b,d\}$ by Lemma \ref{RefineLem}, but  $d\not\in F(\mathbb{M}_3)$ by the Condorcet loser criterion, so $F(\mathbb{M}_3)=\{b\}$.  Then by positive involvement,  $b\in F(\mathbb{M}_2)$, contradicting $F(\mathbb{M}_2)=\{d\}$.

Thus, $F(\mathbb{M}_1)=\{a\}$. Then by positive involvement and resolvability, $F(\mathbb{M}_4)=\{a\}$. But  $F(\mathbb{M}_5)=\{d\}$ by Lemma \ref{RefineLem}, so by positive involvement,  $d\in F(\mathbb{M}_4)$, contradicting $F(\mathbb{M}_4)=\{a\}$.\end{proof}

Thus, in the search for the ``holy grail'' of a voting method satisfying positive involvement, the Condorcet winner and loser criteria, and resolvability, we must drop the restriction of ordinal margin invariance.\footnote{Natural weakenings of ordinal margin invariance to consider are the C1.5 (\citealt{DeDonder2000}) and C2 (\citealt{Fishburn1977}) invariance conditions.} Whether such a method exists or another impossibility theorem awaits us is an important open question.\\

\begin{figure}[h]

\begin{center}

\begin{minipage}{1.75in}
\begin{center}

\begin{tikzpicture}

\node[circle,draw, minimum width=0.25in, fill=gray!75] at (0,0) (a) {$a$}; 
\node[circle,draw,minimum width=0.25in] at (3,0) (c) {$c$}; 
\node[circle,draw,minimum width=0.25in] at (1.5,1.5) (b) {$b$}; 

\node[circle,draw,minimum width=0.25in, fill=gray!75] at (1.5,-1.5) (d) {$d$}; 

\path[->,draw,thick] (a) to[pos=.7] node[fill=white] {$5$} (c);
\path[->,draw,thick] (b) to  node[fill=white] {$4$} (a);
\path[->,draw,thick] (d) to  node[fill=white] {$1$} (a);
\path[->,draw,thick] (d) to  node[fill=white] {$2$} (c);
\path[->,draw,thick] (b) to[pos=.7]   node[fill=white] {$3$} (d);
\path[->,draw,thick] (c) to   node[fill=white] {$6$} (b);
\end{tikzpicture}

$\mathbb{M}_1$

\end{center}
\end{minipage}$\Rightarrow_d$\begin{minipage}{1.75in}
\begin{center}
\begin{tikzpicture}

\node[circle,draw, minimum width=0.25in] at (0,0) (a) {$a$}; 
\node[circle,draw,minimum width=0.25in] at (3,0) (c) {$c$}; 
\node[circle,draw,minimum width=0.25in, fill=gray!75] at (1.5,1.5) (b) {$b$}; 

\node[circle,draw,minimum width=0.25in, fill=gray!75] at (1.5,-1.5) (d) {$d$}; 

\path[->,draw,thick] (a) to[pos=.7] node[fill=white] {$6$} (c);
\path[->,draw,thick] (b) to  node[fill=white] {$5$} (a);
\path[->,draw,thick] (d) to  node[fill=white] {$1$} (a);
\path[->,draw,thick] (d) to  node[fill=white] {$2$} (c);
\path[->,draw,thick] (b) to[pos=.7]   node[fill=white] {$3$} (d);
\path[->,draw,thick] (c) to   node[fill=white] {$4$} (b);
\end{tikzpicture}

$\mathbb{M}_2$
\end{center}
\end{minipage}$\Leftarrow_b$\begin{minipage}{1.75in}
\begin{center}
\begin{tikzpicture}

\node[circle,draw, minimum width=0.25in] at (0,0) (a) {$a$}; 
\node[circle,draw,minimum width=0.25in] at (3,0) (c) {$c$}; 
\node[circle,draw,minimum width=0.25in, fill=gray!75] at (1.5,1.5) (b) {$b$}; 

\node[circle,draw,minimum width=0.25in, fill=gray!75] at (1.5,-1.5) (d) {$d$}; 

\path[->,draw,thick] (a) to[pos=.7] node[fill=white] {$6$} (c);
\path[->,draw,thick] (b) to  node[fill=white] {$5$} (a);
\path[->,draw,thick] (a) to  node[fill=white] {$2$} (d);
\path[->,draw,thick] (c) to  node[fill=white] {$1$} (d);
\path[->,draw,thick] (b) to[pos=.7]   node[fill=white] {$3$} (d);
\path[->,draw,thick] (c) to   node[fill=white] {$4$} (b);
\end{tikzpicture}

$\mathbb{M}_3$
\end{center}
\end{minipage}
\end{center}

\begin{center}
\begin{minipage}{1.75in}
\begin{center}

\begin{tikzpicture}

\node[circle,draw, minimum width=0.25in, fill=gray!75] at (0,0) (a) {$a$}; 
\node[circle,draw,minimum width=0.25in] at (3,0) (c) {$c$}; 
\node[circle,draw,minimum width=0.25in] at (1.5,1.5) (b) {$b$}; 

\node[circle,draw,minimum width=0.25in, fill=gray!75] at (1.5,-1.5) (d) {$d$}; 

\path[->,draw,thick] (a) to[pos=.7] node[fill=white] {$5$} (c);
\path[->,draw,thick] (b) to  node[fill=white] {$4$} (a);
\path[->,draw,thick] (d) to  node[fill=white] {$1$} (a);
\path[->,draw,thick] (d) to  node[fill=white] {$2$} (c);
\path[->,draw,thick] (b) to[pos=.7]   node[fill=white] {$3$} (d);
\path[->,draw,thick] (c) to   node[fill=white] {$6$} (b);
\end{tikzpicture}

$\mathbb{M}_1$

\end{center}\end{minipage}$\Rightarrow_a$\begin{minipage}{1.75in}
\begin{center}
\begin{tikzpicture}

\node[circle,draw, minimum width=0.25in, fill=gray!75] at (0,0) (a) {$a$}; 
\node[circle,draw,minimum width=0.25in] at (3,0) (c) {$c$}; 
\node[circle,draw,minimum width=0.25in] at (1.5,1.5) (b) {$b$}; 

\node[circle,draw,minimum width=0.25in, fill=gray!75] at (1.5,-1.5) (d) {$d$}; 

\path[->,draw,thick] (a) to[pos=.7] node[fill=white] {$5$} (c);
\path[->,draw,thick] (b) to  node[fill=white] {$4$} (a);
\path[->,draw,thick] (d) to  node[fill=white] {$1$} (a);
\path[->,draw,thick] (d) to  node[fill=white] {$3$} (c);
\path[->,draw,thick] (b) to[pos=.7]   node[fill=white] {$2$} (d);
\path[->,draw,thick] (c) to   node[fill=white] {$6$} (b);
\end{tikzpicture}

$\mathbb{M}_4$
\end{center}
\end{minipage}$\Leftarrow_d$\begin{minipage}{1.75in}
\begin{center}
\begin{tikzpicture}

\node[circle,draw, minimum width=0.25in] at (0,0) (a) {$a$}; 
\node[circle,draw,minimum width=0.25in] at (3,0) (c) {$c$}; 
\node[circle,draw,minimum width=0.25in] at (1.5,1.5) (b) {$b$}; 

\node[circle,draw,minimum width=0.25in, fill=gray!75] at (1.5,-1.5) (d) {$d$}; 

\path[->,draw,thick] (a) to[pos=.7] node[fill=white] {$4$} (c);
\path[->,draw,thick] (b) to  node[fill=white] {$6$} (a);
\path[->,draw,thick] (d) to  node[fill=white] {$1$} (a);
\path[->,draw,thick] (d) to  node[fill=white] {$3$} (c);
\path[->,draw,thick] (b) to[pos=.7]   node[fill=white] {$2$} (d);
\path[->,draw,thick] (c) to   node[fill=white] {$5$} (b);
\end{tikzpicture}

$\mathbb{M}_5$
\end{center}
\end{minipage}
\end{center}
\caption{Ordinal margin graphs for the proof of Theorem \ref{MainThm}.}\label{RelMarginsFig}
\end{figure}

\begin{figure}[!htb]
\begin{center}
\begin{minipage}{1.5in}
\begin{center}
\setlength{\tabcolsep}{4pt} 
\begin{tabular}{cccccc}
$8$ & $8$ & $14$ & $6$ & $7$ & $2$\\\hline 
$a$ & $a$ & $b$ & $c$ & $c$ & $d$\\ 
$c$ & $d$ & $d$ & $b$ & $d$ & $c$\\ 
$b$ & $c$ & $a$ & $a$ & $b$ & $b$\\ 
$d$ & $b$ & $c$ & $d$ & $a$ & $a$
\end{tabular}
\end{center}
\begin{center}
$\mathbf{P}_1$
\end{center}

\begin{center}
{\footnotesize

$c$ beats $b$ by $31 - 14 = 17$

$a$ beats $c$ by $30 - 15 = 15$

$b$ beats $a$ by $29 - 16 = 13$

$b$ beats $d$ by $28 - 17 = 11$

$d$ beats $c$ by $24 - 21 = 3$

$d$ beats $a$ by $23 - 22 = 1$

}
\end{center}
\end{minipage} $+$ \begin{minipage}{.4in}
\begin{center}
\begin{tabular}{c}
$3$\\\hline 
$d$\\ 
$b$\\ 
$a$\\ 
$c$
\end{tabular}
\end{center}
\begin{center}
$\mathbf{P}_{1\to 2}$
\end{center}

\begin{center}
{\footnotesize 
$\,$

$\,$

$\,$

$\,$

$\,$

$\,$}
\end{center}
\end{minipage} $=$ \begin{minipage}{1.5in}

\begin{center}
\setlength{\tabcolsep}{4pt} 
\begin{tabular}{ccccccc}
$8$ & $8$ & $14$ & $6$ & $7$ & $3$ & $2$\\\hline 
$a$ & $a$ & $b$ & $c$ & $c$ & $d$ & $d$\\ 
$c$ & $d$ & $d$ & $b$ & $d$ & $b$ & $c$\\ 
$b$ & $c$ & $a$ & $a$ & $b$ & $a$ & $b$\\ 
$d$ & $b$ & $c$ & $d$ & $a$ & $c$ & $a$
\end{tabular}
\end{center}

\begin{center}
$\mathbf{P}_2$
\end{center}

\begin{center}

{\footnotesize

$a$ beats $c$ by $33 - 15 = 18$

$b$ beats $a$ by $32 - 16 = 16$

$c$ beats $b$ by $31 - 17 = 14$

$b$ beats $d$ by $28 - 20 = 8$

$d$ beats $c$ by $27 - 21 = 6$

$d$ beats $a$ by $26 - 22 = 4$

}
\end{center}

\end{minipage}

\end{center}
\caption{Profiles realizing $\mathbb{M}_1$ and $\mathbb{M}_2$ and the transition between them in Figure \ref{RelMarginsFig}.}\label{UpperLeftToMiddleProfs}
\vspace{.2in}
\end{figure}

\begin{figure}[!htb]

\begin{center}
\begin{minipage}{1.5in}
\begin{center}
\setlength{\tabcolsep}{4pt} 
\begin{tabular}{cccccc}
$10$ & $16$ & $1$ & $5$ & $7$ & $16$\\\hline 
$a$ & $b$ & $b$ & $b$ & $d$ & $d$\\ 
$c$ & $a$ & $a$ & $d$ & $a$ & $c$\\ 
$b$ & $c$ & $d$ & $a$ & $c$ & $b$\\ 
$d$ & $d$ & $c$ & $c$ & $b$ & $a$
\end{tabular}
\end{center}
\begin{center}
$\mathbf{Q}_2$
\end{center}

\begin{center}
{\footnotesize

$a$ beats $c$ by $39 - 16 = 23$

$b$ beats $a$ by $38 - 17 = 21$

$c$ beats $b$ by $33 - 22 = 11$

$b$ beats $d$ by $32 - 23 = 9$

$d$ beats $c$ by $29 - 26 = 3$

$d$ beats $a$ by $28 - 27 = 1$

}
\end{center}
\end{minipage} $=$ \begin{minipage}{.4in}
\begin{center}
\begin{tabular}{c}
$4$\\\hline 
$b$\\ 
$d$\\ 
$a$\\ 
$c$
\end{tabular}
\end{center}
\begin{center}
$\mathbf{Q}_{2\leftarrow 3}$
\end{center}

\begin{center}
{\footnotesize 
$\,$

$\,$

$\,$

$\,$

$\,$

$\,$}
\end{center}
\end{minipage} $+$ \begin{minipage}{1.5in}

\begin{center}
\setlength{\tabcolsep}{4pt} 
\begin{tabular}{cccccc}
$10$ & $16$ & $1$ & $1$ & $7$ & $16$\\\hline 
$a$ & $b$ & $b$ & $b$ & $d$ & $d$\\ 
$c$ & $a$ & $a$ & $d$ & $a$ & $c$\\ 
$b$ & $c$ & $d$ & $a$ & $c$ & $b$\\ 
$d$ & $d$ & $c$ & $c$ & $b$ & $a$
\end{tabular}
\end{center}

\begin{center}
$\mathbf{Q}_3$
\end{center}

\begin{center}

{\footnotesize

$a$ beats $c$ by $35 - 16 = 19$

$b$ beats $a$ by $34 - 17 = 17$

$c$ beats $b$ by $33 - 18 = 15$

$b$ beats $d$ by $28 - 23 = 5$

$a$ beats $d$ by $27 - 24 = 3$

$c$ beats $d$ by $26 - 25 = 1$

}
\end{center}

\end{minipage}

\end{center}

\caption{Profiles realizing $\mathbb{M}_2$ and $\mathbb{M}_3$ and the transition between them in Figure \ref{RelMarginsFig}.}\label{UpperRightToMiddleProfs}
\vspace{.2in}
\end{figure}

\begin{figure}[!htb]
\begin{center}
\begin{minipage}{1.5in}
\begin{center}
\setlength{\tabcolsep}{4pt} 
\begin{tabular}{ccccc}
$14$ & $4$ & $8$ & $11$ & $2$\\\hline 
$a$ & $b$ & $b$ & $c$ & $c$\\ 
$d$ & $a$ & $d$ & $b$ & $d$\\ 
$c$ & $c$ & $a$ & $d$ & $b$\\ 
$b$ & $d$ & $c$ & $a$ & $a$
\end{tabular}
\end{center}
\begin{center}
$\mathbf{R}_1$
\end{center}

\begin{center}
{\footnotesize

$c$ beats $b$ by $27 - 12 = 15$

$a$ beats $c$ by $26 - 13 = 13$

$b$ beats $a$ by $25 - 14 = 11$

$b$ beats $d$ by $23 - 16 = 7$

$d$ beats $c$ by $22 - 17 = 5$

$d$ beats $a$ by $21 - 18 = 3$

}
\end{center}
\end{minipage} $+$ \begin{minipage}{.4in}
\begin{center}
\begin{tabular}{c}
$2$\\\hline 
$a$\\ 
$d$\\ 
$c$\\ 
$b$
\end{tabular}
\end{center}
\begin{center}
$\mathbf{R}_{1\to 4}$
\end{center}

\begin{center}
{\footnotesize 
$\,$

$\,$

$\,$

$\,$

$\,$

$\,$}
\end{center}
\end{minipage} $=$ \begin{minipage}{1.5in}

\begin{center}
\setlength{\tabcolsep}{4pt} 
\begin{tabular}{ccccc}
$16$ & $4$ & $8$ & $11$ & $2$\\\hline 
$a$ & $b$ & $b$ & $c$ & $c$\\ 
$d$ & $a$ & $d$ & $b$ & $d$\\ 
$c$ & $c$ & $a$ & $d$ & $b$\\ 
$b$ & $d$ & $c$ & $a$ & $a$
\end{tabular}
\end{center}

\begin{center}
$\mathbf{R}_4$
\end{center}

\begin{center}

{\footnotesize

$c$ beats $b$ by $29 - 12 = 17$

$a$ beats $c$ by $28 - 13 = 15$

$b$ beats $a$ by $25 - 16 = 9$

$d$ beats $c$ by $24 - 17 = 7$

$b$ beats $d$ by $23 - 18 = 5$

$d$ beats $a$ by $21 - 20 = 1$

}
\end{center}

\end{minipage}

\end{center}
\caption{Profiles realizing $\mathbb{M}_1$ and $\mathbb{M}_4$ and the transition between them in Figure \ref{RelMarginsFig}.}\label{LowerLeftToMiddleProfs}
\vspace{.2in}
\end{figure}

\begin{figure}[!htb]

\begin{center}
\begin{minipage}{1.5in}
\begin{center}
\setlength{\tabcolsep}{4pt} 
\begin{tabular}{ccccccc}
$16$ & $9$ & $5$ & $4$ & $3$ & $3$ & $14$\\\hline 
$a$ & $b$ & $b$ & $c$ & $d$ & $d$ & $d$\\ 
$c$ & $a$ & $d$ & $d$ & $a$ & $b$ & $c$\\ 
$b$ & $d$ & $a$ & $b$ & $c$ & $a$ & $b$\\ 
$d$ & $c$ & $c$ & $a$ & $b$ & $c$ & $a$
\end{tabular}
\end{center}
\begin{center}
$\mathbf{S}_4$
\end{center}

\begin{center}
{\footnotesize

$c$ beats $b$ by $37 - 17 = 20$

$a$ beats $c$ by $36 - 18 = 18$

$b$ beats $a$ by $35 - 19 = 16$

$d$ beats $c$ by $34 - 20 = 14$

$b$ beats $d$ by $30 - 24 = 6$

$d$ beats $a$ by $29 - 25 = 4$

}
\end{center}
\end{minipage} $=$ \begin{minipage}{.4in}
\begin{center}
\begin{tabular}{c}
$3$\\\hline 
$d$\\ 
$a$\\ 
$c$\\ 
$b$
\end{tabular}
\end{center}
\begin{center}
$\mathbf{S}_{4\leftarrow 5}$
\end{center}

\begin{center}
{\footnotesize 
$\,$

$\,$

$\,$

$\,$

$\,$

$\,$}
\end{center}
\end{minipage} $+$ \begin{minipage}{1.5in}

\begin{center}
\setlength{\tabcolsep}{4pt} 
\begin{tabular}{cccccc}
$16$ & $9$ & $5$ & $4$ & $3$ & $14$\\\hline 
$a$ & $b$ & $b$ & $c$ & $d$ & $d$\\ 
$c$ & $a$ & $d$ & $d$ & $b$ & $c$\\ 
$b$ & $d$ & $a$ & $b$ & $a$ & $b$\\ 
$d$ & $c$ & $c$ & $a$ & $c$ & $a$
\end{tabular}
\end{center}

\begin{center}
$\mathbf{S}_5$
\end{center}

\begin{center}

{\footnotesize

$b$ beats $a$ by $35 - 16 = 19$\\

$c$ beats $b$ by $34 - 17 = 17$\\

$a$ beats $c$ by $33 - 18 = 15$\\

$d$ beats $c$ by $31 - 20 = 11$\\

$b$ beats $d$ by $30 - 21 = 9$\\

$d$ beats $a$ by $26 - 25 = 1$

}
\end{center}

\end{minipage}

\end{center}

\caption{Profiles realizing $\mathbb{M}_4$ and $\mathbb{M}_5$ and the transition between them in Figure \ref{RelMarginsFig}.}\label{LowerRightToMiddleProfs}
\end{figure}

\newpage

\subsection*{Acknowledgements}

I thank Yifeng Ding, Jobst Heitzig, Milan Moss\'{e}, Eric Pacuit, Zoi Terzopoulou, Nic Tideman, Snow Zhang, Bill~Zwicker, and an anonymous referee for helpful comments.

\newpage

\bibliographystyle{plainnat}

\end{document}